\let\accentvec\vec

\documentclass[runningheads,envcountsect,10pt]{llncs}

\let\vec\accentvec

\usepackage{amsmath}
\usepackage{mathrsfs}

\usepackage{multirow}
\usepackage{amsfonts}
\usepackage[colorlinks]{hyperref}
\usepackage{setspace}
\usepackage{textcomp}
\usepackage{enumerate}
\usepackage[ruled,linesnumbered]{algorithm2e}

\usepackage{graphicx}
\DeclareGraphicsExtensions{.pdf,.jpg,.png}

\usepackage[title,header]{appendix}

\begin{document}
%

\title{A Sub-linear Time Algorithm for Approximating k-Nearest-Neighbor with Full Quality Guarantee}

\titlerunning{Approximate k-NN with Full Approximation Guarantee}
%
\author{Hengzhao Ma\inst{1} \and
	Jianzhong Li\inst{2}
}
\authorrunning{H. Ma, J. Li}
%
\institute{
	\email{hzma@stu.hit.edu.cn}\and
	\email{lijzh@hit.edu.cn}\\
	Harbin Institute of Technology, Harbin, Heilongjiang 150001, China
}
\maketitle              
\begin{abstract}
	In this paper we propose an algorithm for the approximate k-Nearest-Neighbors problem. According to the existing researches, there are two kinds of approximation criterion. One is the distance criteria, and the other is the recall criteria. All former algorithms suffer the problem that there are no theoretical guarantees for the two approximation criterion. The algorithm proposed in this paper unifies the two kinds of approximation criterion, and has full theoretical guarantees. Furthermore, the query time of the algorithm is sub-linear. As far as we know, it is the first algorithm that achieves both sub-linear query time and full theoretical approximation guarantee.
	
\end{abstract}

\section{Introduction}\label{sec:intro}

The k-Nearest-Neighbor (kNN) problem is a well-known problem in theoretical computer science and applications. Let $(U,D)$ be a metric space, then for the input set $P\subseteq U$ of elements and a query element $q\in U$, the kNN problem is to find the $k$ elements with smallest distance to $q$. Since the exact results are expensive to compute when the size of the input is large \cite{Indyk1998}, and approximate results serve as good as the exact ones in many applications \cite{Weber1998}, the approximate kNN, kANN for short, draws more research efforts in recent years. There are two kinds of approximation criterion for the kANN problem, namely, the distance criterion and the recall criterion. The distance criterion requires that the ratio between the distance from the approximate results to the query and the distance from the exact results to the query is no more than a given threshold. The recall criterion requires that the size of the intersection of the approximate result set and the exact result set is no less than a given threshold. The formal description will be given in detail in Section \ref{sec:pre}. Next we brief the existing algorithms for the kANN problem to see how these two criteria are considered by former researchers.

The algorithms for the kANN problem can be categorized into four classes.
The first class is the tree-based methods. The main idea of this method is to recursively partition the metric space into sub-spaces, and organize them into a tree structure. The K-D tree \cite{Bentley1975} is the representative idea in this category. It is efficient in low dimensional spaces, but the performance drops rapidly when the number of dimension grows up. Vantage point tree (VP-tree) \cite{Yianilos1993} is another data structure with a better partition strategy and better performance. The FLANN \cite{Muja2014} method is a recent work with improved performance in high dimensional spaces, but it is reported that this method would achieve in sub-optimal results \cite{Lin2019}. To the best of our knowledge, the tree based methods can satisfy neither the distance nor the recall criterion theoretically.

The second class is the permutation based methods. The idea is to choose a set of pivot points, and represent each data element with a permutation of the pivots sorted by the distance to it. In such a representation, close objects will have similar permutations. Methods using the permutation idea include the MI-File \cite{Amato2014} and PP-Index \cite{Esuli2012}. Unfortunately, the permutation based method can not satisfy either of the distance or the recall criterion theoretically, as far as we know.

The third class is the Locality Sensitive Hashing (LSH) based methods. LSH was first introduced by Indyk et. \cite{Indyk1998} for the kANN problem where $k=1$.  Soon after, Datar et. \cite{Datar2004} proposed the first practical LSH function, and since then there came a burst in the theoretical and applicational researches on the LSH framework. For example, Andoni et. proved the lower bound of the time-space complexities of the LSH based algorithms \cite{Andoni2017}, and devised the optimal LSH function which meets the lower bound \cite{Andoni2015}. On the other hand, Gao et. \cite{Gao2015} proposed an algorithm that aimed to close the gap between the LSH theory and kANN search applications.  See \cite{Wang2014} for a survey. The basic LSH based method can satisfy only the distance criterion when $k=1$ \cite{Indyk1998}. Some existing algorithms made some progress. The C2LSH algorithm \cite{Gan2012} solved the kANN problem with the distance criterion, but it has a constraint that the approximation factor must be a square of an integer. The SRS algorithm \cite{Sun2014} is another one aimed at the distance criterion. However, it only has partial guarantee, that is, the results satisfy the distance criterion only when the algorithm terminates on a specific condition.

The forth class is graph based methods. The specific kind of graphs used in this method is the proximity graphs, where the edges in this kind of graph are defined by the geometric relationship of the points. See \cite{Mitchell2017} for a survey. The graph based kANN algorithms usually conduct a navigating process on the proximity graphs. This process selects an vertex in the graph as the start point, and move to the destination point following some specific navigating strategy. For example, Paredes et. \cite{Paredes2005} used the kNN graph, Ocsa et. \cite{Ocsa2007} used the Relative Neighborhood Graph (RNG), and Malkov et. \cite{Malkov2014} used the Navigable Small World Graph (NSW) \cite{Malkov2014}. None of these algorithms have theoretical guarantee on the two approximation criteria.

In summary, most of the existing algorithms do not have theoretical guarantee on either of the two approximation criteria. The recall criterion is only used as a measurement of the experimental results, and the distance criterion is only partially satisfied by only a few algorithms \cite{Gan2012,Sun2014}.  
In this paper, we propose a sub-linear time algorithm for kANN problem that unifies the two kinds of approximation criteria, which overcomes the disadvantages of the existing algorithms. The contributions of this paper are listed below.

\begin{enumerate}
	\item We propose an algorithm that unifies the distance criterion and the recall criterion for the approximate k-Nearest-Neighbor problem. The result returned by the algorithm can satisfy at least one criterion in any situation. This is a major progress compared to the existing algorithms.
	\item Assuming the input point set follows the spatial Poisson process, the algorithm takes $O(n\log{n})$ time of preprocessing, $O(n\log{n})$ space, and answers a query in $O(dn^{1/d}\log{n}+kn^\rho\log{n})$ time, where $\rho<1$ is a constant.
	\item The algorithm is the first algorithm for kANN that provides theoretical guarantee on both of the approximation criteria, and it is also the first algorithm that achieves sub-linear query time while providing theoretical guarantees. The former works \cite{Gan2012,Sun2014} with partial guarantee both need linear query time.
\end{enumerate}

The rest of this paper is organized as follows. Section \ref{sec:pre} introduces the definition of the problem and some prerequisite knowledge. The detailed algorithm are presented in Section \ref{sec:alg}. Then the time and space complexities are analyzed in Section \ref{sec:analyz}. Finally the conclusion is given in Section \ref{sec:conc}.

\section{Preliminaries}\label{sec:pre}

\subsection{Problem Definitions}
The problem studied in this paper is the approximate k-Nearest-Neighbor problem, which is denoted as kANN for short. In this paper the problem is constrained to the Euclidean space. The input is a set $P$ of points where each $p\in P$ is a d-dimensional vector $(p^{(1)},p^{(2)},\cdots,p^{(n)})$. The distance between two points $p$ and $p'$ is defined by $D(p,p')=\sqrt{\sum\limits_{i=1}^{d}{(p^{(i)}-p'^{(i)}})^2}$, which is the well known Euclidean distance. Before giving the definition of the kANN problem,  we first introduce the exact kNN problem.

\begin{definition}[kNN]\label{def:knn}
	Given the input point set $P\subset R^d$ and a query point $q\in R^d$, define $kNN(q,P)$ to be the set of $k$ points in $P$ that are nearest to $q$. Formally,
	\begin{enumerate}
		\item $kNN(q,P)\subseteq P$, and $|kNN(q,P)|=k$;
		\item $D(p,q)\le D(p',q)$ for $\forall p\in kNN(q,P)$ and $\forall p'\in P\setminus kNN(q,P)$.
	\end{enumerate}
\end{definition}

Next we will give the definition of the approximate kNN. There are two kinds of definitions based on different approximation criteria.

\begin{definition}[$kANN_c$]\label{def:cknn}
	Given the input point set $P\subset R^d$, a query point $q\in R^d$, and a approximation factor $c>1$, find a point set $kANN_c(q,P)$ which satisfies:
	\begin{enumerate}
		\item $kANN_c(q,P)\subseteq P$, and $|kANN_c(q,P)|=k$;
		\item let $T_k(q,P)=\max\limits_{p\in kNN(q,P)}{D(p,q)}$, then $D(p',q)\le c\cdot T_k(q,P)$ holds for $\forall p'\in kANN_c(q,P)$.
	\end{enumerate}
\end{definition}
\begin{remark}
	The second requirement in Definition \ref{def:cknn} is called the distance criterion.
\end{remark}
\begin{definition}[$kANN_\delta$]\label{def:dknn}
		Given the input point set $P\subset R^d$, a query point $q\in R^d$, and a approximation factor $\delta<1$, find a point set $kANN_\delta(q,P)\subseteq P$ which satisfies:
	\begin{enumerate}
		\item $kANN_\delta(q,P)\subseteq P$, and $|kANN_\delta(q,P)|=k$;
		\item $|kANN_\delta(q,P)\cap kNN(q,P)|\ge \delta \cdot k$.
	\end{enumerate}
\end{definition}

\begin{remark}
	If a kANN algorithm returned a set $S$, the value $\frac{|S\cap kNN(q,P)|}{|kNN(q,P)|}$ is usually called the recall of the set $S$. This is widely used in many works to evaluate the quality of the kANN algorithm. Thus we call the second statement in Definition \ref{def:dknn} as the recall criterion.
\end{remark}

Next we give the definition of the problem studied in this paper, which unifies the two different criteria.

\begin{definition}\label{def:cdknn}
Given the input point set $P\subset R^d$, a query point $q\in R^d$, and approximation factors $c>1$ and $\delta<1$, find a point set $kNN_{c,\delta}(q,P)$ which satisfies:
\begin{enumerate}
	\item $kANN_{c,\delta}(q,P)\subseteq P$, and $|kANN_{c,\delta}(q,P)|=k$;
	\item $kANN_{c,\delta}(q,P)$ satisfies at least one of the distance criterion and the recall criterion. Formally,
	either $D(p',q)\le c\cdot T_k(q,P)$ holds for $\forall p'\in kANN_{c,\delta}(q,P)$, or $|kANN_{c,\delta}(q,P)\cap kNN(q,P)|\ge \delta \cdot k$.
\end{enumerate}
\end{definition}

According to Definition \ref{def:cdknn}, the output of the algorithm is required to satisfy one of the two criteria, but not both. It will be our future work to devise an algorithm to satisfy both of the criteria.

In the rest of this section we will introduce some concepts and algorithms that will be used in our proposed algorithm.

\subsection{Minimum Enclosing Spheres}
The D-dimensional spheres is the generalization of the circles in the 2-dimensional case. Let $c$ be the center and $r$ be the radius. A d-dimensional sphere, denoted as $S(c,r)$, is the set $S(c,r)=\{x\in R^d \mid D(x,c)\le r \}$. Note that the boundary is included. If $q\in S(c,r)$ we say that $q$ falls inside sphere $S(c,r)$, or the sphere encloses point $p$. A sphere $S(c,r)$ is said to pass through point $p$ iff $D(c,p)=r$.

Given a set $P$ of points, the minimum enclosing sphere (MES) of $P$, is the d-dimensional sphere enclosing all points in $P$ and has the smallest possible radius. It is known that the MES of a given finite point set in $R^d$ is unique, and can be calculated by a quadratic programming algorithm \cite{Yildirim2008}. Next we introduce the approximate minimum enclosing spheres.

\begin{definition}[AMES]
	Given a set of points $P\subset R^d$ and an approximation factor $\epsilon<1$, the approximate minimum enclosing sphere of $P$, denoted as $AMES(P,\epsilon)$, is a d-dimensional sphere $S(c,r)$ satisfies:
	
	\begin{enumerate}
		\item $p\in S(c,r)$ for $\forall p\in P$;
		\item $r<(1+\epsilon)r^*$, where $r^*$ is the radius of the exact MES of $P$.
	\end{enumerate}
\end{definition}

The following algorithm can calculate the AMES in $O(n/\epsilon^2)$ time, which is given in \cite{Badoiu2003}.

\begin{algorithm}[H]
	\caption{Compute $AMES$}\label{alg:ames}
	\KwIn{a point set $P$, and an approximation factor $\epsilon$.}
	\KwOut{$AMES(P,\epsilon)$}
	$c_0\gets$ an arbitrary point in $P$\;
	
	\For{$i=1$ to $1/\epsilon^2$}{
		$p_i\gets$ the point in $P$ farthest away from $c_{i-1}$\;
		$c_i\gets c_{i-1}+\frac{1}{i}(p_i-c_{i-1})$\;
	}
\end{algorithm}

The following Lemma gives the complexity of Algorithm \ref{alg:ames} .

\begin{lemma}[\cite{Badoiu2003}]\label{lema:ames-time}
	For given $\epsilon$ and $P$ where $|P|=n$, Algorithm \ref{alg:ames} can calculate $AMES(P,\epsilon)$ in $O(n/\epsilon^2)$ time.
\end{lemma}

\subsection{Delaunay Triangulation}
The Delaunay Triangulation (DT) is a fundamental data structure in computation geometry. The definition is given below.

\begin{definition}[DT]
	Given a set of points $P\subset R^d$, the Delaunay Triangulation is a graph $DT(P)=(V,E)$ which satisfies:
	\begin{enumerate}
		\item $V=P$;
		\item for $\forall p,p'\in P$, $(p,p')\in E$ iff there exists a d-dimensional sphere passing through $p$ and $p'$, and no other $p''\in P$ is inside it.
	\end{enumerate}
\end{definition}

The Delaunay Triangulation is a natural dual of the Voronoi diagram. We omit the details about their relationship since it is not the focus of this paper.

There are extensive research works about the Delaunay triangulation. An important problem is to find the expected properties of $DT$ built on random point sets. Here we focus on the 
point sets that follow the spatial Poisson process in d-dimensional Euclidean space. In this model, for any region $\mathcal{R}\subset R^d$, the probability that $\mathcal{R}$ contains $k$ points follows the Poisson distribution. See \cite{PoissonProcessWiki} for more details. We cite one important property of the spatial Poisson process in the following lemma.

\begin{lemma}[\cite{PoissonProcessWiki}]\label{lema:poisson-process}
	Let $S\subset R^d$ be a point set following the spatial Poisson process. Suppose there are two regions $B\subseteq A\subset R^d$. For any point $p\in S$, if $p$ falls inside $A$ then the probability that $p$ falls inside $B$ is the ratio between the volume of $B$ and $A$. Formally, we have 
	$$\Pr[p\in B\mid p\in A]=\frac{volume(B)}{volume(A)}.$$
\end{lemma}

Further, we cite some important properties of the Delaunay triangulation built on point sets which follow the spatial Poisson process.

\begin{lemma}[\cite{Bern1991}]\label{lema:dt-max-degree}
	Let $S\subset R^d$ be a point set following the spatial Poisson process,
	and $\Delta(G)=\max\limits_{p\in V(G)}|\{(p,q)\in E(G)\}|$ be the maximum degree of $G$. Then the expected maximum degree of $DT(S)$ is $O(\log{n}/\log{\log{n}})$.
\end{lemma}

\begin{lemma}[\cite{Buchin2009}]\label{lema:dt-time}
	Let $S\subset R^d$ be a point set following the spatial Poisson process. The expected time to construct $DT(S)$ is $O(n\log{n})$.
\end{lemma}

\subsection{Walking in Delaunay Triangulation}
Given a Delaunay Triangulation $DT$, the points and edges of $DT$ form a set of simplices. Given a query point $q$, there is a problem to find which simplex of $DT$ that $q$ falls in. There is a class of algorithms to tackle this problem which is called Walking. The Walking algorithm start at some simplex, and \emph{walk} to the destination by moving to adjacent simplices step by step. There are several kinds of walking strategy, including Jump$\&$Walk \cite{Mucke1999}, Straight Walk \cite{Bose2007} and Stochastic Walk \cite{Devillers2006}, etc. Some of these strategies are only applicable to 2 or 3 dimensions, while Straight Walk can generalize to higher dimension. As Figure \ref{fig:straight-walk} shows, the Straight Walk strategy only considers the simplices that intersect the line segment from the start point to the destination. The following lemma gives the complexity of this walking strategy. 

\begin{lemma}[\cite{Castro2011}]\label{lema:waking-in-dt} 
	Given a Delaunay Triangulation $DT$ of a point set $P\subset R^d$, and two points $p$ and $p'$ in $R^d$ as the start point and destination point, the walking from $p$ to $p'$ using Straight Walk takes $O(n^{1/d})$ expected time.
\end{lemma}

\begin{figure}\label{fig:straight-walk}
	\centering
		\includegraphics[width=0.8\textwidth]{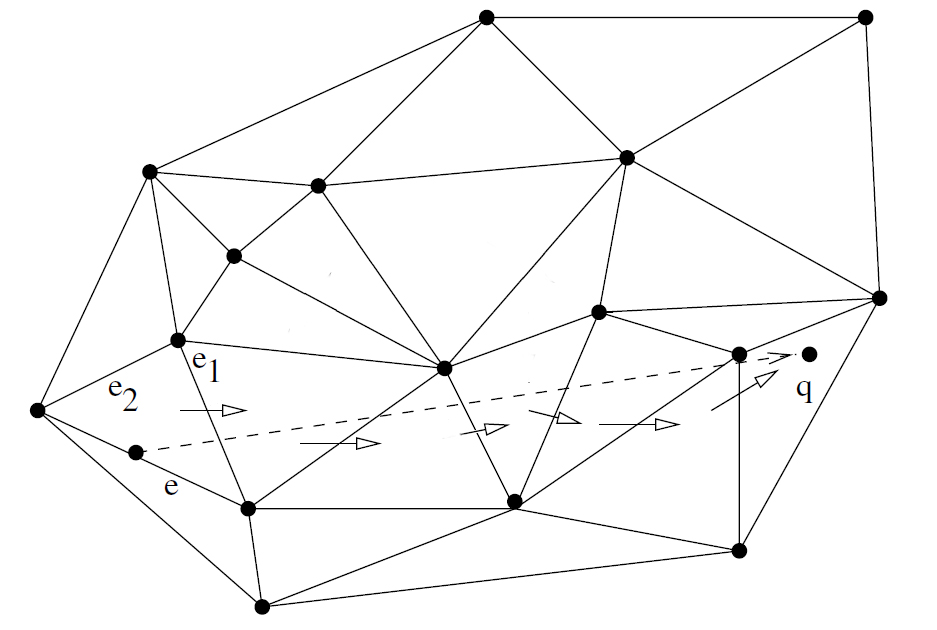}
	\caption{Illustration of the Straight Walk}
\end{figure}

\subsection{$(c,r)$-NN}
The Approximate Near Neighbor problem is introduced in \cite{Indyk1998} for solving the $kANN_c$ problem with $k=1$. Usually the Approximate Near Neighbor problem is denoted as $(c,r)$-NN since there are two input parameters $c$ and $r$. The definition is given below. The idea to use $(c,r)$-NN to solve $1ANN_c$ is via Turing reduction, that is, use $(c,r)$-NN as an oracle or sub-procedure. The details can be found in \cite{Indyk1998,Har-Peled2001,Har-Peled2012,Ma2018}.

\begin{definition}
	Given a point set $P$, a query point $q$, and two query parameters $c>1,r>0$, the output of the $(c,r)$-NN problem should satisfy:
	\begin{enumerate}
		\item if $\exists p^* \in S(q,r)\cap P$, then output a point $p' \in S(q,c\cdot r)\cap P$;
		\item if $D(p,q)> c \cdot r$ for $\forall p\in P$, then output $No$;
	\end{enumerate}
\end{definition}

Since we aim to solve kANN problem in this paper, we need the following definition of $(c,r)$-kNN.

\begin{definition}
	Given a point set $P$, a query point $q$, and two query parameters $c,r$, the output of the $(c,r)$-kNN problem is a set $kNN_{(c,r)}(q,P)$, which satisfies:
	\begin{enumerate}
		\item if $|P\cap S(q,r)|\ge k$, then output a set $Q \subseteq P\cap S(q,c\cdot r)$, where $|Q|=k$;
		\item if $|P\cap S(q,c\cdot r)|< k$, then output $\emptyset$;
	\end{enumerate}
\end{definition}

It can be easily seen that the $(c,r)$-kNN problem is a natural generalization of the $(c,r)$-NN problem. Recently, there are several algorithms proposed to solve this problem. The following Lemma \ref{lema:crknn-complexity} gives the complexity of the $(c,r)$-kNN algorithm, which will be proved in Appendix \ref{apdx:proof-crknn}.

\begin{lemma}\label{lema:crknn-complexity}
	There is an algorithm that solves $(c,r)$-kNN problem in $O(kn^\rho)$ of time, requiring $O(kn^{1+\rho}\log{n})$ time of preprocessing and $O(kn^{1+\rho})$ of space. The parameter $\rho$ is a constant depending on the LSH function used in the algorithm, and $\rho<1$ always holds.
\end{lemma}

\section{Algorithm}\label{sec:alg}
The proposed algorithm consists of two phases, i.e., the preprocessing phase and the query phase. The preprocessing phase is to built a data structure, which will be used to guide the search in the query phase. Next we will describe the algorithm of the two phases in detail.

\subsection{Preprocessing Algorithm}\label{subsec:preprocessing}
Before describing the details of the preprocessing algorithm, we first introduce several concepts that will be used in the following discussion.

\subsubsection{Axis Parallel Box.}
An axis parallel box $B$ in $R^d$ is defined to be the Cartesian product of $d$ intervals, i.e., $B=I_1\times I_2\times\cdots\times I_d$. And the following is the definition of Minimum Bounding Box.

\begin{definition}
	Given a point set $P$, the Minimum Bounding Box, denoted as $MBB(P)$, is the axis parallel box satisfying the following two requirements:
	
	\begin{enumerate}
		\item $MBB(P)$ encloses all points in $P$, and
		\item there exists points $p$ and $p'$ in $P$ such that $p^{(i)}=a_i, p'^{(i)}=b_i$ for each interval $I_i=(a_i,b_i)$ defining $MBB(P)$, $1\le i\le d$.
	\end{enumerate}
\end{definition}

\subsubsection{Median Split}
Given a point set $P$ and its minimum bounding box $MBB(P)$, we introduce an operation on $P$ that splits $P$ into two subsets, which is called median split. This operation first finds the longest interval $I_i$ from the intervals defining $MBB(P)$. Then, the operation finds the median of the set $\{p^{(i)}\mid p\in P \}$, which is the median of the $i$-th coordinates of the points in $P$. This median is denoted as $med_i(P)$. Finally $P$ is split into two subsets, i.e., $P_1=\{p\in P\mid p^{(i)}\le  med_i(P)\}$ and $P_2=\{p\in P\mid p^{(i)}>  med_i(P)\}$. Here we assume that no two points share the same coordinate in any dimension. This assumption can be assured by adding some random small shift on the original coordinates. 

\subsubsection{Median Split Tree}

By recursively conducting the median split operation, a point set $P$ can be organized into a tree structure, which is called the Median Split Tree (MST). The definition of MST is given below. 

\begin{definition}
	Given the input point set $P$, a Median Split Tree (MST) based on $P$, denoted as $MST(P)$, is a tree structure satisfying the following requirements:
	\begin{enumerate}
		\item the root of $MST(P)$ is $P$, and the other nodes in $MST(P)$ are subsets of $P$;
		\item there are two child nodes for each interior node $N\in MST(P)$, which are generated by conducting a median split on $N$;
		\item each leaf node contains only one point.
	\end{enumerate}
\end{definition}

\subsubsection{Balanced Median Split Tree}
The depth of a node $N$ in a tree $T$, denoted as $dep_T(N)$, is defined to be the number of edges in the path from $N$ to the root of $T$. It can be noticed that the leaf nodes in the $MST$ may have different depths. So we introduce the Balanced Median Split Tree (BMST), where all leaf nodes have the same depth. 

Let $L_T(i)=\{N\in T\mid dep_T(N)=i \}$, which is the nodes in the $i$-th layer in tree $T$, and $|N|$ be the number of points included in node $N$. For a median split tree $MST(P)$, it can be easily proved that either $|N|=\lceil n/2^i\rceil$ or $|N|=\lfloor n/2^i\rfloor$ for $\forall N\in L_{MST(P)}(i)$. Given $MST(P)$, the $BMST(P)$ is constructed as follows. Find the smallest $i$ such that $\lfloor n/2^i\rfloor\le 3$, then  for each node $N\in L_{MST(P)}(i)$, all the nodes in the sub-tree rooted at $N$ are directly connected to $N$.

\subsubsection{Hierarchical Delaunay Graph}
Given a point set $P$, we introduce the most important concept for the preprocessing algorithm in this paper, which is the Hierarchical Delaunay Graph (HDG). This structure is constructed by adding edges between nodes in the same layer of $BMST(P)$. The additional edges are called the graph edges, in contrast with the tree edges in $BMST(P)$. The definition of the HDG is given below. Here $Cen(N)$ denotes the center of $AMES(N)$.

\begin{definition}
	Given a point set $P$ and the balanced median split tree $BMST(P)$, a Hierarchical Delaunay Graph $HDG$ is a layered graph based on $BMST(P)$, where each layer is a Delaunay triangulation. Formally, for each $N,N'\in HDG(P)$, there is an graph edge between $N,N'$ iff
	\begin{enumerate}
		\item $dep_{BMST(P)}(N)=dep_{BMST(P)}(N')$, and
		\item  there exists a d-dimensional sphere $S$ passing through $Cen(N),Cen(N')$, and there is no $N''\in HDG(P)$ such that $Cen(N'')$ falls in $S$, where $N''$ is in the same layer with $N$ and $N'$. That is, the graph edges connecting nodes in the same layer forms the Delaunay Triangulation.
	\end{enumerate}
\end{definition}

\begin{algorithm}[t]
	\caption{Preprocessing Algorithm}\label{alg:pre}
	\KwIn{a point set $P$}
	\KwOut{a hierarchical Delaynay graph $HDG(P)$}
	\SetKwFunction{SplitTree}{SplitTree}
	\SetKwFunction{ComputeSpheres}{ComputeSpheres}
	\SetKwFunction{AMES}{AMES}
	\SetKwFunction{Delaunay}{Delaunay}
	\SetKwFunction{HierarchicalDelaunay}{HierarchicalDelaunay}
	\SetKwProg{Procedure}{Procedure}{\string:}{end}
	$T\gets$\SplitTree{P}\;
	Modify $T$ into a BMST\;
	\ComputeSpheres{T}\;
	
	\HierarchicalDelaunay{T}\;
	
	\Procedure{\SplitTree{N}}{
		Conduct median split on $N$ and generate two sets $N_1$ and $N_2$\;
		$T_1\gets$\SplitTree{$N_1$}\;
		$T_2\gets$\SplitTree{$N_2$}\;		
		Let $T_1$ be the left sub-tree of $N$, and $T_2$ be the right sub-tree of $N$;
	}
	\Procedure{\ComputeSpheres{T}}{
		\ForEach{$N\in T$}{
			Call \AMES{$N,0.1$} (Algorithm \ref{alg:ames})\;
		}
	}
	\Procedure{\HierarchicalDelaunay{$T$}}{
		Let $dl$ be the depth of the leaf node in $T$\;
		\For{$i=0$ to $dl$}{
			\Delaunay{$L_T(i)$} (Lemma \ref{lema:dt-time})\;
		}
	}
	
\end{algorithm}

\subsubsection{The preprocessing algorithm}
Next we describe the preprocessing algorithm which aims to build the HDG. The algorithm can be divided into three steps.

Step 1, Split and build tree. The first step is to recursively split $P$ into smaller sets using the median split operation, and the median split tree is built. Finally the nodes near the leaf layer is adjusted to satisfy the definition of the balanced median split tree.

Step 2, Compute Spheres. In this step, the algorithm will go over the tree and compute the AMES for each node using Algorithm \ref{alg:ames}.
 
Step 3, Construct the $HDG$. In this step, an algorithm given in \cite{Buchin2009} which  satisfies Lemma \ref{lema:dt-time} is invoked to compute the Delaunay triangulation for each layer.

The pseudo codes of the preprocessing algorithm is given in Algorithm \ref{alg:pre}.

\subsection{Query Algorithm}
The query algorithm takes the $HDG$ built by the preprocessing algorithm, and executes the following three steps. 

The first is the descending step. The algorithm goes down the tree and stops at level $i$ such that $k\le n/2^i< 2k$. At each level, the child node with smallest distance to the query is chosen to be visited in next level. 

The second is the navigating step. The algorithm marches towards the local nearest AMES center by moving on the edges of the $HDG$. 

The third step is the answering step. The algorithm finds the answer of $kANN_{c,\delta}(q,P)$ by invoking the $(c,r)$-kNN query. The answer can satisfy the distance criterion or the recall criterion according to the different return result of the $(c,r)$-kNN query. 

Algorithm \ref{alg:query} describes the above process in pseudo codes, where $Cen(N)$ and $Rad(N)$ are the center and radius of the $AMES$ of node $N$, respectively.

\begin{algorithm}[t]
	\caption{Query}\label{alg:query}
	\KwIn{a query points $q$, a point set $P$, approximation factors $c>1,\delta <1$, and $HDG(P)$}
	\KwOut{$kANN_{c,\delta}(q,P)$}
	$N\gets$ the root of $HDG(P)$\;
	\While{$|N|>2k$}{
		$Lc\gets$ the left child of $N$, $Rc\gets$ the right child of $N$\;
		\uIf{$D(q,Cen(Lc))<D(q,Cen(Rc)))$}{
			$N\gets Lc$\;
		}\Else{
			$N\gets Rc$\;
		}
	}
	
	\While{$\exists N'\in Nbr(N)$ s.t. $D(q,Cen(N'))<D(q,Cen(N)))$}{
		$N\gets \arg\min\limits_{N'\in Nbr(N)}\{ D(q,Cen(N')) \}$\;
	}
	\For{$i=0$ to $\log_c{n}$}{\label{line:query:for-start}
		Invoke $(c,r)$-kNN query where $r=\frac{D(q,Cen(N))+Rad(N)}{n}c^i$\;
		\If{the query returned a set $Res$}{
			return $Res$ as the final result\;
		}
	}\label{line:query:for-end}
	
\end{algorithm}

\section{Analysis}\label{sec:analyz}
The analysis in this section will assume that the input point set $P$ follows the spatial Poisson process.

\subsection{Correctness}

\begin{lemma}\label{lema:res-d}
	If Algorithm \ref{alg:query} terminates when $i=0$, then the returned point set $Res$ is a $\delta$-kNN of $q$ in $P$ with at least $1- e^{-\frac{n-k}{n^d}}$ probability.
\end{lemma}

\begin{proof}
	Let $R=D(q,Cen(N))+Rad(N)$, $R_0=R/n$, $t\in [0,k]$ be an integer. We define the following three events.
	$$A=\{|P\cap S(q,R_0)|\ge k\} $$
	$$B=\{|P\cap S(q,R_0)|\ge k+t \}$$
	$$C=\{Res\cap kNN(q,P)|\le \delta\cdot k \}$$
	
	The lemma states the situation that the algorithm returns at $i=0$, which implies that event $A$ happens. Event $C$ represents the situation that $Res$ is a $\delta$-kNN set. Then it is easy to see that the desired probability is in this lemma is $1-\Pr[C\mid A]$. By the formula of conditional probability, 
	$$\Pr[C\mid A]=\Pr[C\mid B,A]\Pr[B\mid A]\le \Pr[B\mid A].$$
	Thus in the rest of the proof we focus on calculate $\Pr[B\mid A]$.
	
	To calculate $\Pr[B\mid A]$ we need the probability that a single point $p$ falls in $S(q,R_0)$. We have the following calculations.
	
	\[\begin{aligned}
	\Pr[p\in S(q,R_0)]& = &\Pr[p\in S(q,R_0)\mid p\in S(q,R)]\cdot \Pr[p\in S(q,R)] \\
	& \le& \Pr[p\in S(q,R_0)\mid p\in S(q,R)]\\
	& =& 1/n^d
	\end{aligned}\]

	The last equation is based on Lemma \ref{lema:poisson-process}. 
	
	On the other hand, the number of points in $S(q,R)$ is at most $n$. Here we use the trivial upper bound of $n$ since it is sufficient to the proof. Denote $P=\Pr[p\in S(q,R_0)]$, we have the following equations.
	
	$$\Pr[B\mid A] \le \binom{n-k}{t}P^t(1-P)^{n-k-t}
	 \le e^{-P(n-k)}\frac{P^t(n-k)^t}{t!} $$

	By the property of the Poisson Distribution, the above equation achieves the maximum when $t=\lfloor(n-k)P\rfloor=\lfloor(n-k)/n^d\rfloor=0$. Thus we have $\Pr[B\mid A]\le e^{-\frac{n-k}{n^d}}$. 
	
	Finally, combining the above analysis, we achieve the result that $Res$ is a $\delta$-kNN set with at least $1- e^{-\frac{n-k}{n^d}}$ probability.
	\qed
\end{proof}

\begin{lemma}\label{lema:res-c}
	If Algorithm \ref{alg:query} returns at $i>0$, then the returned point set $Res$ is a $c$-kNN of $q$ in $P$.
\end{lemma}

\begin{proof}
	Let $R_{i-1}=\frac{D(q,Cen(N))+Rad(N)}{n}c^{i-1}$ and $R_i=\frac{D(q,Cen(N))+Rad(N)}{n}c^i$, which are the input parameter of the $(i-1)$-th and $i$-th invocation of the $(c,r)$-kNN query. The lemma states the situation that the algorithm returns at the $i$-th loop, which implies that the $(c,r)$-kNN query returns empty set in the $(i-1)$-th loop. According to the definition of the $(c,r)$-kNN problem, the number of points is less than $k$ in the d-dimensional sphere $S(q,R_{i-1})$. Denote $T_k(q,P)=\max\limits_{p\in kNN(q,P)}{D(q,p)}$, then it can be deduced that $T_k(q,P)\ge R_{i-1}$. On the other hand, the algorithm returns at the $i$-th loop, which implies that the $(c,r)$-kNN query returns a subset of $P\cap S(q,R_i)$. Thus we have $D(p,q)\le R_i$ for each $p$ in the result. Finally, $D(q,p)/T_k\le R_i/R_{i-1}=c$, which exactly satisfies the definition of $c$-kNN.
	\qed
\end{proof}

\begin{theorem}
	The result of Algorithm \ref{alg:query} satisfies the requirement of $kNN_{c,\delta}(q,P)$ with at least $1- e^{-\frac{n-k}{n^d}}$ probability.
\end{theorem}

\begin{proof}
	The result can be directly deduced by combining Lemma \ref{lema:res-d} and \ref{lema:res-c}.
	\qed
\end{proof}

\subsection{Complexities}
For ease of understanding, We first analyze the complexity of the single steps in Algorithm \ref{alg:pre} and \ref{alg:query}.

\begin{lemma}\label{lema:pre-step1}
	The first step of Algorithm \ref{alg:pre} takes $O(n\log{n})$ time.
\end{lemma}

\begin{proof}
	The following recursion formula can be easily deduced from the pseudo codes of Algorithm \ref{alg:pre}.
	$$T(n)=2T(n/2)+O(n)$$
	The $O(n)$ term comes from the time of splitting and computing the median. This recursion formula can be solved by standard process, and the result is $T(n)=O(n\log{n})$.	
	\qed
\end{proof}

\begin{lemma}\label{lema:pre-step2}
	The second step of Algorithm \ref{alg:pre} takes $O(n\log{n})$ time.
\end{lemma}

\begin{proof}
	According to lemma \ref{lema:ames-time}, the time to compute the AMES of a point set is proportional to the number of points in this set. Thus we have the following recursion formula:
	$$T(n)=2T(n/2)+O(n)$$
	The answer is also $T(n)=O(n\log{n})$.
	\qed
\end{proof}

\begin{lemma}\label{lema:pre-step3}
	The third step of Algorithm \ref{alg:pre} takes $O(n\log{n})$ time.
\end{lemma}

\begin{proof}
	According to the definition and the building process of the $HDG$, there are $2^i$ nodes in the $i$-th layer. And by Lemma \ref{lema:dt-time}, the time to build the Delaunay triangulation is $O(n\log{n})$. Thus, the time complexity of the third step is represented by the following equation.
	$$\sum\limits_{i=0}^{\log{n}} {2^i\log{2^i}}=\sum\limits_{i=0}^{\log{n}} {i\cdot 2^i}$$
	
	This result of this additive equation is $O(n\log{n})$.
	\qed
\end{proof}

\begin{lemma}\label{lema:query-navi}
	The navigating step (second step) in Algorithm \ref{alg:query} needs $O(dn^{1/d}\log{n})$ time.
\end{lemma}

\begin{proof}
	From Lemma \ref{lema:waking-in-dt} we know that the navigating step passes at most $O(n^{1/d})$ simplices. A d-dimensional simplex has $d-1$ vertexes, and thus the number of points passed by the navigation process is $O(dn^{1/d})$.
	While a node is visited, the process goes over the neighbors of this node, and from Lemma \ref{lema:dt-max-degree} we know that the expected maximum degree of each node is $O(\log{n})$. Thus the total time of the navigating is $O(dn^{1/d}\log{n})$.
\end{proof}

Now we are ready to present the final results about the complexities.

\begin{theorem}
	The expected time complexity of Algorithm \ref{alg:pre}, which is the preprocessing time complexity, is $O(n\log{n})$.
\end{theorem}

\begin{proof}
	The preprocessing consists of three steps, the time complexities of which are shown in Lemma \ref{lema:pre-step1}, \ref{lema:pre-step2} and \ref{lema:pre-step3}. Adding them and we get the desired conclusion.
\end{proof}

\begin{theorem}
	The space complexity of Algorithm \ref{alg:pre} is $O(n\log{n})$.
\end{theorem}

\begin{proof}
	The space needed to store the $HDG$ is the proportional to the number of the graph edges and tree edges. The number of graph edges connected to each node is $O(\log{n})$ according to Lemma \ref{lema:dt-max-degree}, and the number of tree edges is constant. On the other hand, the number of nodes in $HDG$ is $O(n)$. Finally, we get the result that the space complexity is $O(n\log{n})$.
	\qed
\end{proof}

\begin{theorem}
	The time complexity of Algorithm \ref{alg:query}, which is the query complexity, is $O(dn^{1/d}\log{n}+kn^\rho\log{n})$, where $\rho<1$ is a constant.
\end{theorem}

\begin{proof}
	The time complexity of Algorithm \ref{alg:query} consists of three parts. For the first part, which is descending part, it is easy to see that the complexity is $O(\log{(n/k)})$. And the time complexity of the second part is already solved by Lemma \ref{lema:query-navi}, which is $O(dn^{1/d}\log{n})$. The third part is to invoke the $(c,r)$-kNN query for $\log{n}$ times. By Lemma \ref{lema:crknn-complexity} each invocation of $(c,r)$-kNN needs $O(kn^\rho)$ where $\rho>1$ is a constant.  If we take this algorithm as a Turing reduction, then the time to invoke $(c,r)$-kNN query is constant. Thus the third step needs $kn^\rho\log{n}$ time.
	Adding the three parts and the desired result is achieved.
	\qed
\end{proof}

\section{Conclusion}\label{sec:conc}
In this paper we proposed an algorithm for the approximate k-Nearest-Neighbors problem. We observed that there are two kinds of approximation criterion in the history of this research area, which is called the distance criteria and the recall criteria in this paper. But we also observed that all existing works do not have theoretical guarantees on this criteria. We raised a new definition for the approximate k-Nearest-Neighbor problem which unifies the distance criteria and the recall criteria, and proposed an algorithm that solves the new problem. The result of the algorithm can satisfy at least one of the two criterion. In our future work, we will try to devise new algorithms that can satisfy both of the criterion.

\bibliographystyle{splncs04}
\bibliography{library}

\begin{subappendices}
		\renewcommand{\thesection}{\Alph{section}}
	\section{Proof of Lemma \ref{lema:crknn-complexity}}\label{apdx:proof-crknn}
	\begin{proof}
		The algorithm for $(c,r)$-kNN is adapted from the standard LSH algorithm for $(c,r)$-NN. See \cite{Datar2004} for more details. Briefly speaking, let $\mathcal{H}$ be a family of LSH functions, ${h_i^j}$ be a set of LSH functions uniformly drawn from $\mathcal{H}$, $1\le i\le M, 1\le j\le L$, and $\mathcal{G}_j(p)=(h_1^j(p),\cdots,h_M^j)$ be a composition of $M$ LSH functions. The algorithm stores each element $p$ in the input point set $P$ in the hash bucket $\mathcal{G}_j(p)$, $1\le j\le L$, and for the query point $q$, the algorithm scans the buckets $\mathcal{G}_j{q}$, $1\le j\le L$, and collects the points in $S(q,r)$. If the algorithm collects $k$ points in $S(q,cr)$, then the algorithm returns the $k$ points. If the algorithm have scanned $3L$ points before collects  enough points, it returns $No$.
		
		Now we prove the algorithm succeeds with constant probability. The algorithm succeeds if the following two conditions are true. Here we call the points out side $S(q,cr)$ as \emph{outer} points.
		
		\begin{enumerate}[A.]
			\item If there exists $p_1,\cdots,p_k\in B(q,r)$, then for each $p_i$ there exists $\mathcal{G}_j$ such that $\mathcal{G}_j(p_i)=\mathcal{G}_j(q)$, and
			\item the algorithm encounters at most $3L$ outer points.
		
		\end{enumerate}
		First we introduce the following two probabilities.

		Let $P_1=\Pr[\mathcal{G}_j(p')=\mathcal{G}_j(q)\mid D(p',q)\ge cr]$. Apparently $P_1\le p_2^K$. 
		Then let $K=\log_{1/p_2}{n}\Rightarrow P_1\le\frac{1}{n}$. For all points outside $B(q,cr)$, the expected number of outer points satisfying $\mathcal{G}_j(p')=\mathcal{G}_j(q)$ for some $j$ is at most $\frac{1}{n}\times n-1$.		
		Thus for all $1\le i\le L$, the expected number of outer points satisfying $\mathcal{G}_j(p')=\mathcal{G}_j(q)$ is at most $L$.
		
		By Markov's inequality, 
		$$\Pr[|\{\mathcal{G}_j(p')=\mathcal{G}_j(q) \mid D(p',q)\ge cr \}|\ge 3L]\le L/3L=\frac{1}{3}.$$
		This is	the possibility of scanning at least $3L$ outer points, which is the possibility that event $B$ fails. Thus $\Pr[B]\ge \frac{2}{3}$.
		
		On the other hand, let $P_2=\Pr[\mathcal{G}_j(p)=\mathcal{G}_j(q)\mid D(p,q)\le r]$. By setting $M=\log_{1/p_2}{n}$ we have the following derivations
		
		$P_2\ge p_1^{M}=p_1^{\log_{1/p_2}{n}}=n^{-\frac{\log_{1/p_1}{n}}{\log_{1/p_2}{n}}}=n^{-\rho}$
		
		Setting $L=kn^{\rho}$, the possibility that there exists at least one $1\le j\le L$  such that $\mathcal{G}_j(p)=\mathcal{G}_j(q)$ is at least 
		$$1-(1-P_2)^L\ge 1-(1-n^{-\rho})^{kn^{\rho}}\ge 1-e^{-k}.$$
		
		For all the $k$ points, the possibility that $p$ coincides with each of the $k$ points, which is $\Pr[A]$, is that 
		$$\Pr[A]\ge(1-e^{-k})^k\ge 1-e^{-1}$$
		
		The last inequality comes from the montonicity of the function $(1-e^{-x})^x$.
		Finally, the probability of condition $A$ and $B$ both succeed can be easily computed, which is constant probability. The details are omitted.
		
		After all, we have proved that the algorithm can return the result of $(c,r)$-kNN with constant probability by setting $M=\log_{1/p_2}{n}$ and $L=kn^{\rho}$. Finally substituting the $M$ and $L$ with the proper values, we get the complexities of the algorithm, which is stated in Lemma \ref{lema:crknn-complexity}.
		
	\end{proof}
		
\end{subappendices}

\end{document}